\documentclass[11pt]{article}
\usepackage[margin=1.0in]{geometry}

\usepackage{color}
\usepackage{url}
\usepackage{times}
\usepackage{xspace}
\usepackage{listings}
\usepackage{enumitem}
\usepackage{amssymb}
\usepackage{amsmath}
\usepackage[normalem]{ulem}

\usepackage{caption}

\usepackage{authblk}

\newcommand{\E}{{\mathbb E}}

\newcommand{\SL}{\mbox{SL}\xspace}
\newcommand{\LL}{\mbox{LL}\xspace}

\newcommand{\Alg}{\mbox{Alg}\xspace}
\newcommand{\OPT}{\mbox{OPT}\xspace}
\newcommand{\OFF}{\mbox{OFF}\xspace}
\newcommand{\ep}{\varepsilon}

\newcommand{\rhoflr}{\overline{\gamma}}
\newcommand{\rhoover}{\hat{\gamma}}

\renewcommand{\rho}{\gamma}

\newcommand{\lmin}{{\ell_{min}}\xspace}
\newcommand{\lmax}{{\ell_{max}}\xspace}
\newcommand{\pmin}{{p}\xspace}
\newcommand{\pmax}{{q}\xspace}

\newcommand{\Ralg}{\mbox{SL-Preamble}\xspace}
\newcommand{\CRalg}{\mbox{CSL-Preamble}\xspace}

\newtheorem{proposition}{Proposition}
\newtheorem{theorem}{Theorem}
\newtheorem{lemma}{Lemma}

\newenvironment{proof}{\noindent{\bf Proof:}}{\hfill\rule{2mm}{2mm}\\}    

\newcommand{\sq}{\hbox{\rlap{$\sqcap$}$\sqcup$}}
\newcommand{\qed}{\hspace*{\fill}\sq}

\newenvironment{packed_enum}{
\begin{enumerate}
  \setlength{\itemsep}{1pt}
  \setlength{\parskip}{0pt}
  \setlength{\parsep}{0pt}
}{\end{enumerate}
}

\newcommand{\remove}[1]{}

\usepackage[final]{changes}
\definechangesauthor{af}{red}
\definechangesauthor{cg}{blue}
\definechangesauthor{dk}{brown}
\definechangesauthor{ez}{magenta}
\begin{document}

\title{Measuring the Impact of Adversarial Errors\\ on Packet Scheduling Strategies\thanks{This research was supported in part by the Comunidad de Madrid grant S2009TIC-1692, Spanish MICINN/MINECO grant TEC2011-29688-C02-01, and NSF of China grant 61020106002.}
}

\author[1]{Antonio Fern\'andez Anta}
\author[2]{Chryssis Georgiou}
\author[3]{Dariusz R. Kowalski}
\author[1]{Joerg Widmer}
\author[1,4]{Elli Zavou\thanks{Partially supported by FPU Grant from MECD}}

\affil[1]{Institute IMDEA Networks}
\affil[2]{University of Cyprus}
\affil[3]{University of Liverpool}
\affil[4]{Universidad Carlos III de Madrid}

\date{}

\maketitle
\begin{abstract}
In this paper we explore the problem of achieving efficient packet 
transmission over unreliable links with worst case occurrence of 
errors. In such a setup, even an omniscient offline scheduling strategy cannot achieve stability of the packet queue, nor is it able to use up all the available bandwidth. Hence, an important 
first step is to identify an appropriate metric for measuring the 
efficiency of scheduling strategies in such a setting. To this end, we propose a \emph{relative throughput} metric which corresponds to the \emph{long term competitive ratio} of the algorithm 
with respect to the optimal. We then explore the impact of the error detection mechanism and feedback delay 
on our measure. We compare instantaneous  
error feedback with deferred error feedback, that requires a 
faulty packet to be fully received in order to detect the 
error. We propose algorithms for worst-case adversarial and stochastic packet arrival models, and formally analyze their performance. The relative throughput achieved by these algorithms is shown to be close to optimal by deriving lower bounds on the relative throughput of the algorithms and almost matching upper bounds for any algorithm in the considered settings.
Our collection of results demonstrate the potential of using instantaneous feedback to improve the performance of communication systems in adverse environments.
\end{abstract}

%\keywords{ACM proceedings, \LaTeX, text tagging}

%!TEX root = ./main.tex

\section{Introduction}
\label{s:intro}

\paragraph{\bf Motivation.}

Packet scheduling~\cite{packet_scheduling} is one of the most fundamental problems in computer networks. 
As packets arrive, the sender (or scheduler) needs to continuously make scheduling decisions. 
Typically,
the objective is 
to maximize the {\em throughput} of the link or to achieve stability. 
Furthermore, the sender needs to take decisions without knowledge of future packet
arrivals. Therefore, many times this problem is treated as an {\em online} scheduling
problem~\cite{Peleg1992,online_scheduling} and {\em competitive analysis}~\cite{aspnes,tarjan} is used to evaluate the performance of
proposed solutions: the worst-case performance of an online algorithm is compared with the
performance of an offline optimal algorithm that has a priori knowledge of the problem's
input. \vspace{-.1em}

In this work we focus on online packet scheduling over {\em unreliable} links,
where packets transmitted over the link might be corrupted by bit errors. 
Such errors may, for example, be caused by an increased noise level or transient interference on the link, that in the worst case could be caused by a malicious entity or an attacker. 
In the case of an error the affected packets must be retransmitted.
To investigate 
the impact of such errors on the scheduling problem under study and provide
{\em provable guarantees}, we consider the worst case occurrence of errors, that is, we consider errors 
caused by an omniscient and adaptive {\em adversary}~\cite{jamming}. The adversary has full knowledge of the 
protocol and its history, and it uses this knowledge to decide whether it will cause errors on 
the packets transmitted in the link at a certain time or not. 
Within this general framework, the packet arrival is continuous and can either be controlled by the adversary or be
stochastic.

\begin{table}
\begin{footnotesize}
\begin{center}
	\begin{tabular}{| c | c | c | c |}
	\hline 
	 Arrivals & Feedback & Upper Bound & Lower Bound \\ \hline \hline
	 & Deferred  & $0$ & $0$ \\ \cline{2-4} 
	 & & & \\	 
	Adversarial &  Instantaneous & $T_{\Alg} \leq \rhoflr/(\rho + \rhoflr)$ & $T_{SL-Pr} \geq \rhoflr/(\rho + \rhoflr)$ \\
	  &  &  & \\
	 &  &  $T_{LL} = 0$, $T_{SL} \leq 1/(\rho + 1)$ & \\ \hline \hline
	 
	 & Deferred & $0$ & $0$ \\ \cline{2-4}
	 & & & \\	 
	Stochastic & Instantaneous& $T_{\Alg} \leq \rhoflr/\rho$ & 
	$T_{CSL-Pr} \geq \rhoflr/(\rho + \rhoflr)$, if $\lambda p \lmin \le \rhoflr/ (2\rho)$ \\
	 & & $T_{\Alg} \leq \max\left\{ \lambda p\lmin, \rhoflr/(\rho + \rhoflr)\right\}$, if $p<q$ & $T_{CSL-Pr} \geq \min\left\{ \lambda p\lmin, \rhoflr/\rho \right\}$, otherwise \\
	 & & $T_{LL} = 0$, $T_{SL} \leq 1 / (\rho + 1) $  & \\ \hline 	 	
	\end{tabular}
	\caption{Summary of results presented. The results for deferred feedback are for one packet length, while the results for instantaneous feedback are for 2 packet lengths $\lmin$ and $\lmax$. Note that $\rho = \lmax / \lmin$, $\rhoflr = \lfloor\rho\rfloor$, $\lambda p$ is the arrival rate of $\lmin$ packets, and $p$ and $q=1-p$ are the proportions of $\lmin$ and $\lmax$ packets, respectively.
	}
	\label{t:1}
\end{center}
\end{footnotesize}
\end{table}

\paragraph{\bf Contributions.}
Packet scheduling performance is often evaluated using throughput, measured in absolute terms (e.g., in bits per second) or normalized with respect to the bandwidth (maximum transmission capacity) of the link. This throughput metric makes sense for a link without errors or with random errors, where the full capacity of the link can be achieved under certain conditions. However, if adversarial bit errors
can occur during the transmission of packets, the full capacity is usually not achievable by any protocol, unless restrictions are imposed on the adversary~\cite{HSWD,jamming}. Moreover, since a bit error renders a whole packet unusable (unless costly techniques like PPR~\cite{Jamieson:2007:PPP:1282380.1282426} are used), a throughput equal to the capacity minus the bits with errors is not achievable either. 
As a consequence, in a link with adversarial bit errors, a fair comparison should compare the throughput of a specific algorithm to the maximum achievable amount of traffic that \emph{any} protocol could send across the link. This introduces the challenge of identifying an appropriate metric to measure the throughput of a protocol over a link with adversarial errors.\vspace{.3em}

\noindent\emph{Relative throughput:}
Our first contribution is the proposal of a {\em relative throughput} metric
for packet scheduling algorithms under unreliable links (Section~\ref{sec:model}). 
This metric is a variation of the competitive ratio typically considered in online scheduling. 
Instead of considering the ratio of the performance of a given algorithm  over that of the optimal offline algorithm, we consider the limit of this ratio as time goes to infinity. This
corresponds to the {\em long term competitive ratio} of the algorithm with respect to the optimal.\vspace{.3em}

\noindent\emph{Problem outline:}
We consider a sender that transmits packets to a receiver over an unreliable link, where the errors are controlled by an adversary. 
Regarding packet arrivals (at the sender), we consider two models: (a) the arrival times and their sizes follow a stochastic distribution, 
and (b) the arrival times and their sizes are also controlled by an adversary.
The general offline version of our scheduling problem, in which the scheduling algorithm knows a priori when errors 
will occur, is NP-hard \footnote{Some of the results are omitted due to space limitation and can be found in the Appendix.}.
This further motivates the need for devising simple and efficient online algorithms for the problem we consider.\vspace{.3em}

\noindent\emph{Feedback mechanisms:}
Then, moving to the online problem requires detecting %and notifying 
the packets received with errors, in order to retransmit them. 
The usual mechanism~\cite{ARQ_CRC}, which we call {\em deferred feedback}, detects and notifies the sender that a packet has suffered an error after the whole packet has been received by the receiver. 
It can be shown that, 
even when the packet arrivals are stochastic and packets have the same length, no online scheduling algorithm with deferred feedback can 
be competitive with respect to the offline one.
Hence, we center our study a second mechanism, which we call {\em instantaneous feedback}. It detects and notifies the sender of an error the moment this error occurs. 
This mechanism can be thought of as an abstraction of the emerging
Continuous Error Detection (CED) framework~\cite{CED} that uses arithmetic coding to provide continuous error detection.
The difference between deferred and instantaneous feedback is drastic, since for the instantaneous feedback mechanism, and for packets of the same length, it is easy to obtain optimal relative throughput of 1, even in the case of adversarial arrivals.
However, the problem becomes substantially more challenging in the case of non-uniform packet lengths. Hence, we analyze the problem for the case of packets with two different lengths, $\lmin$ and $\lmax$, where $\lmin<\lmax$. \vspace{.3em} 

\noindent\emph{Bounds for adversarial arrivals:}
We show (Section~\ref{sec:adversarial}), that an online algorithm with instantaneous feedback can achieve at most almost half the relative throughput with respect to the offline one.
It can also be shown
that two basic scheduling policies, 
giving 
priority either to short {\em(SL -- Shortest Length)} or long
{\em(LL -- Longest Length)} packets, are not efficient under adversarial errors.
Therefore, 
we devise a new algorithm, called \Ralg, and show that it achieves the optimal online relative throughput. 
Our algorithm, transmits a ``sufficiently'' large number of short packets while making sure that long packets are transmitted from time~to~time. \vspace{.3em}

\noindent\emph{Bounds for stochastic arrivals:}
In the case of stochastic packet arrivals (Section~\ref{sec:stochastic}), as one might expect, we obtain better relative throughput in some cases. The results are summarized in Table~\ref{t:1}.
We propose and analyze an algorithm, called \CRalg, 
that achieves relative throughput that is optimal.
This algorithm schedules packets according to $\Ralg$,
giving preference to short packets depending on the parameters of the stochastic distribution of packet arrivals\setcounter{footnote}{0}\footnote{If the distribution is not known, then obviously one needs to use the algorithm developed for the case of adversarial arrivals that needs no knowledge a priori.}.
We show that the performance of algorithm \CRalg\ is 
optimal for a wide range of parameters of stochastic distributions of packets arrivals, by proving the matching upper bound\footnote{Analyzing algorithms yields lower bounds on the relative throughput, while analyzing adversarial strategies yields upper bounds on the relative throughput.}
for the relative throughput of any algorithm in this setting.\vspace{.3em} 

\noindent\emph{A note on randomization:}
All the proposed algorithms are deterministic. 
Interestingly, it can be shown
that using randomization 
does not improve the results;
the upper bounds already discussed hold also for the randomized case. \added[ez]{For more details see Appendix~\ref{sec:randomization}.}
\vspace{.3em}

\noindent To the best of our knowledge, this is the first work that investigates in depth the impact of adversarial worst-case link errors on the throughput 
of the packet scheduling problem.
Collectively, our results (see Table~\ref{t:1}) show that instantaneous feedback can achieve a significant relative throughput under worst-case adversarial errors (almost half the relative throughput that the offline optimal algorithm can achieve).
Furthermore, we observe that in some cases, stochastic arrivals allow for better performance.

\paragraph{\bf Related work.}
A vast amount of work exists for online (packet) scheduling.
Here we focus only on the work that is most related to 
ours.
For more information the reader can consult~\cite{scheduling} and~\cite{online_scheduling}.
The work in~\cite{kesselheim} considers the packet scheduling problem in
wireless networks. Like our work, it looks at both stochastic and
adversarial arrivals. Unlike our work though, it considers only {\em reliable} links. 
Its main objective is to achieve maximal throughput guaranteeing {\em stabiliy}, meaning bounded time from injection to delivery.
The work in~\cite{HSWD} considers online packet scheduling over
a wireless channel, where both the channel conditions and the data
arrivals are governed by an adversary.
Its main objective is to
design scheduling algorithms for the base-station to achieve stability
in terms of the size of queues of each mobile user.
Our work does not focus on stability, as we assume errors controlled by an unbounded adversary that
can always prevent it. 
The work in~\cite{jamming} %(and earlier work on the same topic)
considers the problem of devising local access control protocols for wireless
networks with a single channel, that are provably robust against {\em adaptive
adversarial jamming}. At certain time steps, the adversary can jam the communication in
the channel in such a way that the wireless nodes do not receive messages
(unlike our work, where the receiver might receive a message, but it
might contain bit errors). Although the model and the objectives of this 
line of work
is different from ours, it shares the same concept of studying the
impact of adversarial behavior on network communication.

\section{Model}
\label{sec:model}

\paragraph{\bf Network setting.} 
We consider a sending station transmitting packets over a link. Packets arrive at the sending station continuously
and may have different lengths.
Each packet that arrives is associated with a length and its arrival time (based on the station's local clock). We denote by $\lmin$ and $\lmax$ the smallest and largest lengths, respectively, that a packet may have.
We use the notation $\rho = \lmax/\lmin$, $\rhoflr = \lfloor\rho\rfloor$ and $\rhoover = \lceil\rho\rceil - 1$.
The link is unreliable, that is, transmitted packets might be corrupted by bit errors. We assume that all packets are transmitted at the same bit rate, hence the transmission time is proportional to the packet's length.

\vspace{-.5em}
\paragraph{\bf Arrival models.}
We consider two models for packet arrivals. 
\begin{itemize}[leftmargin=5mm]
\item {\em Adversarial:}
The packets' arrival time and length are governed by an adversary. We define an adversarial arrival pattern as a collection of packet arrivals caused by the adversary. 
\vspace{.2em}
\item {\em Stochastic:}
We consider a probabilistic distribution $D_a$, under which packets arrive at the sending station and a probabilistic distribution $D_s$, for the length of the packets.
In particular, we assume 
packets arriving according to a Poisson process with parameter $\lambda>0$. When considering two packet lengths, $\lmin$ and $\lmax$, each packet that arrives is assigned one of the two lengths independently, with probabilities $\pmin > 0$ and $\pmax > 0$ respectively, where $\pmin+\pmax=1$.
\end{itemize}

\paragraph{\bf Packet bit errors.}
We consider an adversary that controls the bit errors of the packets transmitted over the link. An adversarial error pattern is defined as a collection of error events on the link caused by the adversary.
More precisely, an error event at time $t$ specifies that an instantaneous error occurs on the link at time $t$, so the packet that happens to be on the link at that time is corrupted with bit errors.
A corrupted packet transmission is unsuccessful, therefore the packet needs
to be retransmitted in full. 
As mentioned before, we consider an {\em instantaneous feedback} mechanism for the notification of the sender about the error. The instant the packet suffers a bit error the sending station is notified (and hence it can stop transmitting the remainder of the packet -- if any). 

\vspace{-.5em}
\paragraph{\bf The power of the adversary.}
Adversarial models are typically used to argue about the algorithm's behavior in worst-case scenarios.
In this work we assume an adaptive adversary that knows the algorithm and 
the history of the execution up to the current point in time.
In the case of stochastic arrivals, this includes all stochastic packet arrivals up to this point, and the length of the packets that have arrived.
However it only knows the distribution but neither 
the exact timing nor the length of the packets arriving beyond the current time.

Note that in the case of deterministic algorithms,
in the model of adversarial arrivals the adversary has full knowledge of the computation,
as it controls both packet arrivals and errors, and can simulate the behavior of the algorithm
in the future (there are no random bits involved in the computation).
This is not the case in the model with stochastic arrivals, where the adversary does not control the timing of future packet arrivals, but knows only about the packet arrival and length distributions.

\paragraph{\bf Efficiency metric: \em Relative throughput.} 
Due to dynamic packet arrivals and adversarial errors, the real link capacity may vary throughout the execution.
Therefore, we view the problem of packet scheduling in this setting as an online problem and we pursue long-term competitive analysis. 
Specifically,
let $A$ be an arrival pattern and $E$ an error pattern. 
For a given deterministic algorithm $\Alg$, 
let $L_{\Alg}(A,E,t)$ be the total length of all the successfully transferred (i.e., non-corrupted) packets by time $t$
under patterns $A$ and $E$.
Let $\OPT$ be the offline optimal algorithm that knows the exact arrival and error patterns before the start of the execution.
We assume that \OPT devises an optimal schedule that maximizes at each time $t$ the successfully transferred packets $L_{\OPT}(A,E,t)$.
Observe that, in the case of stochastic arrivals, the worst-case adversarial error pattern may depend on stochastic injections. Therefore, we view $E$ as a function of an arrival pattern $A$ and time $t$.
In particular, for an arrival pattern $A$ we consider a function $E(A,t)$ that defines errors at time $t$ based on the behavior of a given algorithm \Alg under the arrival pattern $A$ up to time $t$ and the values of function $E(A,t')$ for $t'<t$.

Let ${\mathcal A}$ denote a considered arrival model, i.e., a set of arrival patterns in case of adversarial, or a distribution of packet injection patterns in case of stochastic,
and let ${\mathcal E}$ denote the corresponding adversarial error model, i.e., a set of error patterns derived by the adversary, or a set of functions defining the error event times in response to the arrivals that already took place in case of stochastic arrivals.
In case of adversarial arrivals, we require that any pair of patterns $A\in {\mathcal A}$ and $E\in {\mathcal E}$ occurring in an execution must allow non-trivial communication, i.e., the value of $L_{\OPT}(A,E,t)$ in the execution is unbounded with $t$ going to infinity.
In case of stochastic arrivals, we require that any adversarial error function $E\in {\mathcal E}$ applied in an execution must allow non-trivial communication for any stochastic arrival pattern $A\in {\mathcal A}$.

\sloppy For arrival pattern $A$, adversarial error function $E$ and time $t$, we define the \emph{relative throughput $T_{\Alg}(A,E,t)$ of a deterministic algorithm $\Alg$ by time $t$} as: % follows: 
\[
T_{\Alg}(A,E,t) = \frac{L_{\Alg}(A,E,t)}{L_{\OPT}(A,E,t)}
\ .
\]
For completeness, $T_{\Alg}(A,E,t)$ equals %is defined as
1 if $L_{\Alg}(A,E,t)=L_{\OPT}(A,E,t)=0$.

We define the {\em relative throughput} of algorithm $\Alg$
in the adversarial arrival model as: % follows:
\[
T_{\Alg} = \inf_{A \in {\mathcal A}, E  \in {\mathcal E}} \lim_{t \rightarrow \infty} T_{\Alg}(A,E,t)
\ ,
\]
while in the stochastic arrival model it needs to take into account the random distribution of
arrival patterns in ${\mathcal A}$, and is defined as follows:
\[
T_{\Alg} = \inf_{E  \in {\mathcal E}} \lim_{t \rightarrow \infty} \E_{A \in {\mathcal A}} [T_{\Alg}(A,E,t)]
\ .
\]

To prove lower bounds on relative throughput, we compare the performance of a given algorithm with that of \OPT. 
When deriving upper bounds, it is not necessary to compare the performance of a given algorithm with that
of OPT, but instead, with the performance of some carefully chosen offline algorithm \OFF. 
As we demonstrate later, this approach leads to accurate
upper bound results.

Finally, we consider {\em work conserving} online scheduling algorithms, in the following sense:
as long as there are pending packets, the sender does not cease to schedule packets.
Note that it does not make any difference whether one assumes that offline algorithms are work-conserving or not, since their throughput is the same in both cases (a work conserving offline algorithm always transmits, but stops the ongoing transmission as soon as an error occurs and then continues with the next packet). Hence for simplicity we do not assume offline algorithms to be work conserving.
%\added[ez]{COMMENT: The $3^{rd}$ reviewer mentions that the results should not rely on the adversary being non-work-conserving because it would constitute an unfair advantage in the model.}
%\added[cg]{I believe that by assuming this it makes our analysis worst case and hence our algorithmic solutions even stronger. If it is indeed the case, then we can either mention something like that, or ignore the comment.} 

%\input{nphardness} %MOVED TO APPENDIX_FINAL

\section{Adversarial Arrivals}
\label{sec:adversarial}

This section focuses on adversarial packet arrivals.
First, observe that it is relatively easy and efficient to handle packets of only one length.

\begin{proposition}
\label{t:inst-1len}
Any work conserving online scheduling algorithm with instantaneous feedback has optimal relative throughput of 1 when all packets have the same length.
\end{proposition}

%\ez{BROUGHT BACK FROM APPENDIX}

\begin{proof}
Consider an algorithm \Alg. Since it is work conserving, as long
as there are pending packets, it schedules them. If an error is reported by
the feedback mechanism, the algorithm simply retransmits another (or the same) packet. Since the
notification is instantaneous, it is not difficult to see that the a priori knowledge that 
the offline optimal algorithm has, does not help in transmitting more non-corrupted packets than \Alg.
\end{proof}

\subsection{Upper Bound}
\label{sec:advUB}

Let $\Alg$ be any deterministic algorithm for the considered packet scheduling problem. 
In order to prove upper bounds, \Alg will be competing with an offline algorithm $\OFF$. 
The scenario is as follows.
We consider an infinite supply of packets of length $\lmax$ and initially assume that there are no packets of length $\lmin$.
We define as a {\em link error event}, the point in time when the adversary corrupts (causes an error to) any packet that happens to be in  the link at that specific time. We divide the execution in \emph{phases}, defined as the periods between two consecutive link error events.
We distinguish 2 types of phases as described below and give a description for the behavior of the adversarial models $\mathcal{A}$ and $\mathcal{E}$. 
The adversary controls the arrivals of packets at the sending station and error events of the link, as well as the actions of algorithm $\OFF$.
The two types of phases are as follows:
\begin{enumerate}[leftmargin=5mm]
\item a phase in which $\Alg$ starts by transmitting an $\lmax$ packet (the first phase of the execution belongs to this class). Immediately after $\Alg$ starts transmitting the $\lmax$ packet, a set of $\rhoover$ $\lmin$-packets arrive, that are scheduled and transmitted by $\OFF$. After $\OFF$ completes the transmission of these packets, a link error occurs, so $\Alg$ cannot complete the transmission of the $\lmax$ packet (more precisely, the packet undergoes a bit error, so it needs to be retransmitted). Here we use the fact that $\rhoover < \rho$.
\item a phase in which $\Alg$ starts by transmitting an $\lmin$ packet. In this case, $\OFF$ transmits an $\lmax$ packet. Immediately after this transmission is completed,
a link error occurs. Observe that in this phase $\Alg$ has transmitted successfully several $\lmin$ packets (up to $\rhoflr$ of them).
\end{enumerate}

Let $A$ and $E$ be the specific adversarial arrival and error patterns in an execution of \Alg.
Let us consider any time $t$ (at the end of a phase for simplicity) in the execution. Let $p_1$ be the number of phases of type 1 executed by time $t$. Similarly, let $p_2(j)$ be the number of phases of type 2 executed by time $t$ in which \Alg transmits $j$ $\lmin$ packets, for $j \in [1,\rhoflr]$. Then, the relative throughput can be computed as follows.
\begin{eqnarray}
T_{\Alg}(A,E,t)= 
\frac{ \lmin \sum_{j=1}^{\rhoflr} j p_2(j) }
{\lmax \sum_{j=1}^{\rhoflr} p_2(j) + \lmin \rhoover p_1 } \cdot
\label{eq-rel-thr}
\end{eqnarray}

From the arrival pattern $A$, the number of $\lmin$ packets injected by time $t$ is exactly $\rhoover p_1$. Hence, $\sum_{j=1}^{\rhoflr} j p_2(j) \leq \rhoover p_1$.
It can be easily observed from Eq.~\ref{eq-rel-thr} that the \added[ez]{relative} throughput increases with the average number of $\lmin$ packets transmitted in the phases of type 2. Hence, the throughput would be maximal if all the $\lmin$ packets are used in phases of type 2 with $\rhoflr$ packets. With the above we obtain the following theorem.

\begin{theorem}
\label{t:inst-adv-upper}
The relative throughput of $\Alg$
under adversarial patterns $A$ and $E$ and up to time $t$ is at most  $\frac{\rhoflr}{\rho + \rhoflr} \le \frac{1}{2}$ (the equality holds iff $\rho$ is an integer).
\end{theorem}

\begin{proof}
Applying the bound $\sum_{j=1}^{\rhoflr} p_2(j) \geq \sum_{j=1}^{\rhoflr} \frac{j p_2(j)}{\rhoflr}$ in Eq (1), we get
\begin{equation*}
T_{\Alg}(A,E,t) \leq \frac{\lmin \sum_{j=1}^{\rhoflr} j p_2(j)}{\frac{\lmax}{\rhoflr} \sum_{j=1}^{\rhoflr} j p_2(j) + \lmin\rhoover p_1},
\end{equation*}
which is a function that increases with $\sum_{j=1}^{\rhoflr} j p_2(j)$.
Since
$\sum_{j=1}^{\rhoflr} j p_2(j) \leq \rhoover p_1$,
the relative throughput can be bounded as
\begin{eqnarray*}
T_{\Alg}(A,E,t) & \leq &
\frac{ \lmin \rhoflr \rhoover p_1/ \rhoflr }
{\lmax \frac{\rhoover p_1}{\rhoflr} + \lmin \rhoover p_1 } 
=  \frac{ \lmin \rhoflr }{\lmax  + \lmin  \rhoflr}
= \frac{\rhoflr}{\rho + \rhoflr}.
\end{eqnarray*}
\end{proof}

\subsection{Lower Bound and \Ralg Algorithm}
\label{sec:ralg}

Two natural scheduling policies one could consider are the \emph{Shortest Length} ($\SL$) and
\emph{Longest Length} ($\LL$) algorithms; the first gives priority to
$\lmin$ packets, whereas the second gives priority to the $\lmax$ packets.
However, these two policies are not efficient in the considered setting; $LL$ cannot achieve a relative throughput more than $0$ while $SL$ achieves at most $T = \frac{1}{\rho + 1}$.
Therefore, we present algorithm \Ralg that tries to combine, in a graceful and
efficient manner, these two policies.

\vspace*{-1.5ex}
\paragraph{Algorithm description:}
At the beginning of the execution and whenever the sender is (immediately) notified by the instantaneous feedback mechanism that a link error occurred, 
it checks the queue of pending packets to see whether there are at least $\rhoflr$ packets of length $\lmin$ available for transmission. If there are, then it schedules $\rhoflr$ of them --- this is called a \emph{preamble} --- and then the algorithm continues to schedule packets using the $\LL$ policy.  Otherwise, if there are not enough $\lmin$ packets available, it simply schedules packets following the $\LL$ policy.
\vspace*{1.5ex}

%\ez{BROUGHT FROM APPENDIX}

\vspace*{-1.5ex}
\paragraph{Algorithm analysis:}
We show that algorithm \Ralg achieves a relative throughput that matches the upper bound shown in the previous subsection, and hence, it is optimal.
Let us define two types of time periods for the link in the executions of algorithm \Ralg: the \emph{active} and the \emph{inactive} periods. An active period is one in which the link experiences no errors and \Ralg has pending packets waiting to be transferred, whereas an inactive one is such that either the link has an error point or the queue of pending packets is empty for \Ralg. 
In the case of inactive periods, note that, if the link has an error, neither \Ralg nor \OPT can make any progress in transmitting an error-free packet. Similarly, if the queue of pending packets is empty for \Ralg, it must be empty for \OPT as well (otherwise it would contradict the optimality of \OPT).
Hence, we look at the active periods, which we refer to as \emph{phases}, and according to the above algorithm we observe that there are four types of phases that may occur.
\begin{packed_enum}
\item Phase starting with $\lmin$ packet and has length $L < \rhoflr \lmin$
\item Phase starting with $\lmin$ packet and length $L \geq \rhoflr \lmin$
\item Phase starting with $\lmax$ packet and has length $L < \lmax$
\item Phase starting with $\lmax$ packet and length $L \geq \lmax$
\end{packed_enum}
We now introduce some notation that will be used throughout the analysis.
For the execution of \Ralg and within the $i$th phase, let $a_i$ be the number of successfully transmitted $\lmin$ packets not in the preambles, $b_i$ the number of successfully transmitted $\lmax$ packets, and $c_i$ the number of successfully transmitted $\lmin$ packets in preambles.
For the execution of \OPT and within the $i$th phase, let $a_i^*$ be the total number of successfully transmitted $\lmin$ packets and $b_i^*$ the total number of successfully transmitted $\lmax$ packets.
Let $C_A^j(i)$ and $C_O^j(i)$ denote the total amount successfully transmitted within a phase $i$ of type $j$ by \Ralg and \OPT, respectively.

Analyzing the different types of phases we make some observations. First, for phases of type 1, \Ralg is not able to transmit successfully the $\rhoflr$  $\lmin$ packets of the preamble, but \OPT is only able to complete at most as much work, so $C_O^1 \leq C_A^1$. For phases of type 2, we observe that the amount of work completed by \OPT minus the work completed by \Ralg is at most $\lmax$ (i.e., $C_O^2 - C_A^2 < \lmax$). 
Therefore, $C_O^2 \leq \frac{\lmin \rhoflr}{\lmax+\lmin \rhoflr} C_A^2$. (Observe that $\frac{\lmin \rhoflr}{\lmax+\lmin \rhoflr} \leq 1/2$.)
The same holds for phases of type 4 ($C_O^4 - C_A^4 < \lmax$) and hence in this case  $C_O^4 \leq 2C_A^4$. In the case of phases of type 3, \Ralg is not able to transmit successfully any packet, and therefore $C_A^3=0$, whereas \OPT might transmit up to $\rhoover \lmin$ packets.

\noindent There are two cases of executions to be considered separately.

\noindent{\bf Case 1:} The number of phases of type 3 is finite.\\
\noindent In such a case, there is a phase $i^*$ such that $\forall i>i^*$ phase $i$ is not of type 3. Then
\begin{equation}
R_1 = \frac{\sum\limits_{j\leq i^*}C_A(j) + \sum\limits_{j>i^*}C_A(j)}{\sum\limits_{j\leq i^*}C_O(j) + \sum\limits_{j>i^*}C_O(j)}
\end{equation}

It is clear that the total progress completed by the end of phase $i^*$ by both algorithms is bounded. So we define $\sum\limits_{j\leq i^*}C_A(j) = A$ and $\sum\limits_{j\leq i^*}C_O(j) = O$ and thus,
\begin{equation*}
R_1 = \frac{A + \sum\limits_{j>i^*}C_A(j)}{O + \sum\limits_{j>i^*}C_O(j)} \geq
\frac{A + \frac{\lmin \rhoflr}{\lmax+\lmin \rhoflr} \sum\limits_{j>i^*}C_O(j)}{O + \sum\limits_{j>i^*}C_O(j)}
\end{equation*}
Hence, the relative throughput of \Ralg at the end of each phase, can be computed as $T = \lim_{t \rightarrow \infty}R_1$, i.e.,
\begin{eqnarray*}
T & = &
\lim_{j \rightarrow \infty} \frac{A + \frac{\lmin \rhoflr}{\lmax+\lmin \rhoflr} \sum\limits_{j>i^*}C_O(j)}{O + \sum\limits_{j>i^*}C_O(j)} \\
& = &
\lim_{j \rightarrow \infty} \frac{(\lmax+\lmin \rhoflr)A + (\lmin \rhoflr) \sum\limits_{j>i^*}C_O(j)}{(\lmax+\lmin \rhoflr)(O + \sum\limits_{j>i^*}C_O(j))} \\
&= & 
\lim_{j \rightarrow \infty}(\frac{\lmin \rhoflr}{\lmax+\lmin \rhoflr} + \frac{(\lmax+\lmin \rhoflr)A - (\lmin \rhoflr)O}{(\lmax+\lmin \rhoflr)(O + \sum\limits_{j>i^*}C_O(j))})\\
&= &
\frac{\lmin \rhoflr}{\lmax+\lmin \rhoflr} = \frac{\rhoflr}{\rho + \rhoflr}
\end{eqnarray*}

Here it is important to note that the assumption $\lim_{t \rightarrow \infty}C_O(t) = \infty$ is used, which corresponds to the expression $\lim_{j \rightarrow \infty}\sum\limits_{j>i^*}C_O(j)$ in the above equality.

So far, we have basically seen what is the relative throughput of \Ralg at the end of each phase. It is also important to guarantee the lower bound at all times within the phases.
Consider any time-point $t$ of phase $i>i^*$. Then $R_i(t) = \frac{\sum_{j \in(i^*, i-1]}C_A(j) + X_t}{\sum_{j \in(i^*, i-1]}C_O(j) + Y_t}$, where $X_t$ and $Y_t$ is the work completed by \Ralg and \OPT within phase $i$ up to time $t$. Using our proof above and the fact that for phases of type 1, 2 and 4 $C_A \geq \frac{\lmin \rhoflr}{\lmax+\lmin \rhoflr} C_O$, we know that $X_t \geq \frac{\lmin \rhoflr}{\lmax+\lmin \rhoflr} Y_t$ as well. Therefore,
\begin{eqnarray*}
R_i(t) & \geq & \frac{\frac{\lmin \rhoflr}{\lmax+\lmin \rhoflr} \sum_{j \in(i^*, i-1]}C_O(j) + \frac{\lmin \rhoflr}{\lmax+\lmin \rhoflr} Y_t}{\sum_{j \in(i^*, i-1]}C_O(j) + Y_t}\\
& = & \frac{\lmin \rhoflr}{\lmax+\lmin \rhoflr}
\end{eqnarray*}
This completes the lower bound of relative throughput for Case 1.\\

\noindent{\bf Case 2:} The number of phases of type 3 is infinite.\\
\noindent In this case we must see how the number of $\lmin$ and $\lmax$ packets are bounded for both \Ralg and \OPT.

\begin{lemma}
\label{l:lmin}
Consider the time point $t$ at the beginning of a phase $j$ of type 3. Then the number of $\lmin$ tasks completed by $t$ by \OPT is no more than the amount of $\lmin$ tasks completed by \Ralg plus $\rhoflr - 1$, i.e., $\sum_{i<j}a_i^* \leq \sum_{i<j}(a_i + c_i) + (\rhoflr - 1)$.
\end{lemma}
\begin{proof}
Consider the beginning of phase $j$ of type 3. At that point, we know that \Ralg has at most $(\rhoflr - 1)$ $\lmin$ tasks in its queue of pending tasks by definition of phase type 3. Therefore, the amount of $\lmin$ tasks completed by \OPT by the beginning of phase $j$ is no more than the ones completed by \Ralg (including the $\lmin$ tasks in preambles) plus $\rhoflr - 1$.
\end{proof}

\begin{lemma}
\label{l:lmax}
Considering all kinds of phases and the number of $\lmax$ tasks, 
$\sum\limits_{i\leq j}b_i^* \leq \sum\limits_{i\leq j}b_i + \sum\limits_{i\leq j}\frac{c_i}{\rhoflr} + 2, \forall j$
\end{lemma}
\begin{proof}
We prove this claim by induction on phase $j$.
For the \emph{Base Case: $j=0$} the claim is trivial.
We consider the \emph{Induction Hypothesis} stating that $\sum\limits_{i\leq j-1}b_i^* \leq \sum\limits_{i\leq j-1}b_i + \sum\limits_{i\leq j-1}\frac{c_i}{\rhoflr} + 2$. 
For the \emph{Induction Step} we need to prove it up to the end of phase $j$.
We first consider the case where during the phase $j$ there is a time when \Ralg has no $\lmax$ tasks. Let $t$ be the latest such time in the phase. Let us define $b^*(t)$ and $b(t)$ being the number of $\lmax$ tasks completed up to time $t$ by \OPT and \Ralg respectively. We know that $b^*(t)\leq b(t)$. Let also $x_j^*(t)$ and $x_j(t)$ be the number of $\lmax$ tasks completed by \OPT and \Ralg, respectively, after time point $t$ until the end of the phase $j$.
We claim that $x_j^*(t)\leq x_j(t)+2$. From our definitions, at time $t$ \Ralg is executing a $\lmin$ task. Since $t$ is the last time that \Ralg has no $\lmax$ tasks, the worst case is being at the beginning of the preamble (by inspection of the 4 types of phases). Then, if the phase ends at time $t^\prime$, we define period $I = [t,t^\prime]$:
\begin{eqnarray*}
|I| & < &
\rhoflr\lmin + (x_j(t) + 1)\lmax \\
& \leq &
(x_j(t) + 2)\lmax
\end{eqnarray*}

The +1 $\lmax$ task is because of the crash before completing the last $\lmax$ scheduled task of the phase.
Observe that \OPT could be executing a $\lmax$ task at time $t$, completed at some point in $[t,t+\lmax]$ and accounted for in $x_j^*(t)$. Therefore, 
$$\sum\limits_{i\leq j}b_i^* = b^*(t) + x_j^*(t) \leq b(t) + x_j(t) + 2 = \sum\limits_{i\leq j}b_i + 2.$$

Now consider the case where at all times of a phase $j$ there are $\lmax$ tasks in the queue of \Ralg. 
By inspection of the 4 types of phases, the worst case is when $j$ is of type 2. 
Since there is always some $\lmax$ task pending in \Ralg, after completing the $\rhoflr\lmin$ tasks it will keep scheduling $\lmax$ tasks, until a crash stops the last one scheduled, or the queue becomes empty. On the same time \OPT is able to complete at most $\lfloor\frac{L_j}{\lmax}\rfloor \leq b_j + 1$ $\lmax$-tasks, where $L_j$ is the length of the phase. Therefore, in all types of phases, $b_j^* \leq \frac{c_j}{\rhoflr} + b_j$.
And hence by induction the claim follows; $\sum\limits_{i\leq j}b_i^* \leq \sum\limits_{i\leq j}\frac{c_i}{\rhoflr} + \sum\limits_{i\leq j}b_i + 2$.
\end{proof}

Combining the two lemmas above, Lemma~\ref{l:lmin} and \ref{l:lmax}:
\begin{equation*}
R_2 = \frac{\sum\limits_{i\leq j}C_A(i)}{\sum\limits_{i\leq j}C_O(j)} =
\frac{\sum\limits_{i\leq j}[(a_i+c_i)\lmin + b_i\lmax]}{\sum\limits_{i\leq j}[a_i^*\lmin + b_i^*\lmax]}
\end{equation*}
\begin{eqnarray*}
 &\geq &
\frac{\sum\limits_{i\leq j}[(a_i+c_i)\lmin + b_i\lmax]}{\sum\limits_{i\leq j}(a_i\!+\!c_i)\lmin \!+\! (\rhoflr \!-\! 1)\lmin \!+\! \sum\limits_{i\leq j}(b_i \!+\! \frac{c_i}{\rhoflr})\lmax \!+\! 2\lmax} \\
&\geq &
\frac{\sum\limits_{i\leq j}[(a_i+c_i)\lmin + b_i\lmax]}{\sum\limits_{i\leq j}[(a_i+2c_i)\lmin + b_i\lmax] + 3\lmax} \\
& \geq &
\frac{\sum\limits_{i\leq j}[(a_i+c_i)\lmin + b_i\lmax] + \frac{3}{2}\lmax - \frac{3}{2}\lmax}{2\sum\limits_{i\leq j}[(a_i+c_i)\lmin + b_i\lmax] + 3\lmax} \\
&\geq &
\frac{1}{2} - \frac{\frac{3}{2}\lmax}{2\sum\limits_{i\leq j}[(a_i+c_i)\lmin + b_i\lmax] + 3\lmax} \\
\end{eqnarray*}

Therefore,
\begin{equation}
T = \lim_{j \rightarrow\infty}R_2 \geq \frac{1}{2}
\end{equation}

%\ez{COMPLETE PROVE WAS ADDED - SKETCH NOT NECESSARY}
\remove{
\paragraph{Algorithm analysis (sketch):}
We show that algorithm \Ralg achieves a relative throughput that 
matches the upper bound shown in the previous subsection, and hence, it is optimal.
The line of argumentation is as follows.
According to the algorithm there are four types of phases that may occur.
\begin{packed_enum}
\item Phase starting with $\lmin$ packet and has length $L < \rhoflr \lmin$
\item Phase starting with $\lmin$ packet and length $L \geq \rhoflr \lmin$
\item Phase starting with $\lmax$ packet and has length $L < \lmax$
\item Phase starting with $\lmax$ packet and length $L \geq \lmax$
\end{packed_enum}
For phases of type 1, \Ralg is not able to transmit successfully the $\rhoflr$  packets $\lmin$ of the preamble, but clearly \OPT is only able to complete at most as much {\em work} (understood as the total length of sent packets). 
For phases of type 2 and 4, the amount of work completed by \OPT can be at most the work completed by \Ralg plus $\lmax$ (and hence the
former is at most twice the latter).
In the case of phases of type 3, \Ralg is not able to successfully transmit any packet, whereas \OPT might transmit up to $\rhoover \lmin$ packets.
Amortizing the work completed by \OPT in these phases with those completed in the preambles 
of types~1 and~2 by algorithm \Ralg is the most challenging part of the proof.
This process is divided into two cases, depending on whether the number of type 3 phases is bounded or not.
The details of the analysis of algorithm \Ralg, leading to the following theorem, are included in the Appendix~\ref{a:adv-inst}.
} %end remove

\begin{theorem}
\label{thm:Ralg}
The relative throughput of Algorithm \Ralg is at least $\frac{\rhoflr}{\rho+\rhoflr}$.
\end{theorem}

\begin{proof}
From the analyses of Cases 1 and 2 and the fact that 
$\frac{\rhoflr}{\rho + \rhoflr} \leq \frac{1}{2}$
it is easy to conclude that the relative throughput of Algorithm \Ralg is at least $\frac{\rhoflr}{\rho + \rhoflr}$ as claimed.
\end{proof}

\section{Stochastic Arrivals}
\label{sec:stochastic}

We now turn our attention to stochastic packet arrivals.\vspace{-1em} 

\subsection{Upper Bounds}
\label{sec:stocUB}

In order to find the upper bound of the relative throughput, we consider again an arbitrary work conserving algorithm $\Alg$. 
Recall that we assume that $\lambda p > 0$ and $\lambda q > 0$, which implies that there are in fact injections of 
packets of both lengths $\lmin$ and $\lmax$ (recall the definitions of $\lambda$, $p$ and $q$ from Section~\ref{sec:model}). 
We define the following adversarial error model $\mathcal{E}$.

\vspace*{-1.5ex}
\begin{enumerate}[leftmargin=5mm]
\item 
When $\Alg$ starts a phase by transmitting an $\lmax$ packet then,
\begin{enumerate}[leftmargin=5mm]
\item
\label{1a}
If $\OFF$ has $\lmin$ packets pending, then the adversary extends the phase so that \OFF can transmit successfully as many $\lmin$ packets as possible, up to $\rhoover$.
Then, it ends the phase so that $\Alg$ does not complete the transmission of the $\lmax$ packet (since $\rhoover \lmin < \lmax$).\vspace{.2em}
\item
\label{1b}
If $\OFF$ does not have any $\lmin$ packets pending, then the adversary inserts a link error immediately (say after infinitesimally small time $\epsilon$).\vspace{.2em}
\end{enumerate}

\item 
When $\Alg$ starts a phase by transmitting an $\lmin$ packet then,
\begin{enumerate}[leftmargin=5mm]
\item 
\label{2a}
IF $\OFF$ has a packet of length $\lmax$ pending, then the adversary extends the phase so \OFF can transmit an $\lmax$ packet. By the time this packet is successfully transmitted,
the adversary inserts an error and finishes the phase. Observe that in this case \Alg was able to successfully transmit up to $\rhoflr$ packets $\lmin$.\vspace{.2em} 
\item 
\label{2b}
If $\OFF$ has no $\lmax$ packets pending, then the adversary inserts an error immediately and ends the phase.
\end{enumerate}
\end{enumerate}

\vspace*{-1.5ex}
Observe that in phases of type \ref{1b} and \ref{2b}, neither \OFF nor \Alg are able to transmit any packet. These phases are just used by the adversary to wait for the conditions required by phases of type \ref{1a} and \ref{2a} to hold. In these latter types some packets are successfully transmitted (at least by \OFF). Hence we call them \emph{productive} phases.
Analyzing a possible execution, in addition to the concept of phase that we have already used, we define \emph{rounds}. There is a round associated with each productive phase. The round ends
when its corresponding productive phase ends, and starts at the end of the prior round (or at the start of the execution if no prior round exists). Depending on the type of productive phase
they contain, rounds can be classified as type \ref{1a} or \ref{2a}. 

\newcommand{\ronea}{r_{\ref{1a}}^{(j)}}
\newcommand{\rtwoa}{r_{\ref{2a}}^{(j)}}
\newcommand{\rtwoap}[1]{r_{\ref{2a}}^{(#1)}}

Let us fix some (large) time $t$. We denote by $\ronea$ the number of rounds of type \ref{1a} in which $j \leq \rhoover$ packets of length $\lmin$ are sent by \OFF completed by time $t$. The value $\rtwoa$ with  $j \leq \rhoflr$ packets of length $\lmin$ sent by \Alg, is defined similarly for rounds of type \ref{2a}. (Here rounding effects do not have any significant impact,
since they
will be compensated by the assumption that $t$ is large.) We assume that $t$ is a time when a round finishes. Let us denote by $r$ the total number or rounds
completed by time $t$, i.e., $\sum_{j=1}^{\rhoflr} \rtwoa + \sum_{j=1}^{\rhoover} \ronea =r$.

The relative throughput by time $t$ can be computed as
\begin{eqnarray}
T_{\Alg}(A,E,t)= 
\frac{ \lmin \sum_{j=1}^{\rhoflr} j \cdot \rtwoa }
{\lmax \sum_{j=1}^{\rhoflr} \rtwoa + \lmin \sum_{j=1}^{\rhoover} j \cdot \ronea }
\ .
\label{eq-relative-thr}
\end{eqnarray}

From this expression, we can show the following result.

\begin{theorem}
\label{theo:abs-ub}
No algorithm \Alg has relative throughput larger than 
$\frac{\rhoflr}{\rho}$.
\end{theorem}

\begin{proof}
It can be observed in Eq.~\ref{eq-relative-thr} that, for a fixed $r$, the lower the value of $\ronea$ the higher the relative throughput. 
Regarding the values $\rtwoa$, the throughput increases when 
there are more rounds in the larger values of $j$.
E.g., under the same conditions, a configuration with $\rtwoa = k_1$ and $\rtwoap{j+1} = k_2$, has lower throughput than one with $\rtwoa = k_1-1$ and $\rtwoap{j+1} = k_2+1$.
Then, the throughput is maximized when
$\rtwoap{\rhoflr}=r$ and the rest of values $\ronea$ and $\rtwoa$ are 0, which yields the bound.
\end{proof}

To provide tighter bounds for some special cases, we prove the following lemma.

\begin{lemma}
\label{lem:chernoff}
\label{l:stoch-lmin}
Consider any two constants $\eta,\eta'$ such that
$0 < \eta < \lambda < \eta'$.
Then: 
\vspace{-1.5ex}
\begin{itemize}
\item [(a)]
there is a constant $c>0$, 
dependent only on $\lambda,p,\eta$, such that
for any time $t\ge \lmin$,
the number of packets of length $\lmin$ (resp., $\lmax$) injected by time $t$
is {\em at least} $t\eta p$ (resp., $t\eta q$) with probability at least 
$1-e^{-ct}$;\vspace{.2em}
\item [(b)]
there is a constant $c'>0$, 
dependent only on $\lambda,p,\eta'$, such that
for any time $t\ge \lmin$,
the number of packets of length $\lmin$ (resp., $\lmax$) injected by time $t$
is {\em at most} $t\eta' p$ (resp., $t\eta' q$) with probability at least 
$1-e^{-c't}$.
\end{itemize}
\end{lemma}

\begin{proof}
We first prove the statement 1(a).
The Poisson process governing arrival times of packets of length $\lmin$ has parameter $\lambda p$.
By the definition of a Poisson process, 
%governing packet arrival and independent choice of packet length, 
the distribution of packets of length $\lmin$ arriving to the system in the period $[0,t]$ 
is the Poisson distribution
with parameter $\lambda p t$. Consequently, by Chernoff bound for Poisson random variables
(with parameter $\lambda p t$), c.f., 
\cite{chernoff}, 
the probability that at least $\eta p t$ packets arrive to the system in the period $[0,t]$ 
is at least
\[
1-e^{-\lambda p t} \frac{(e\lambda p t)^{\eta p t}}{(\eta p t)^{\eta p t}} = 
1-e^{-t p (\lambda-\eta\ln(e\lambda/\eta))} \ge 1-e^{-ct}, 
\]
for some constant $c>0$ dependent on $\lambda,\eta,p$.
In the above, the argument behind the last inequality is as follows.
It is a well-known fact that $x>1+\ln x$ holds for any $x>1$; in particular,
for $x=\lambda/\eta>1$. This implies that 
$x-\ln(ex)$ is a positive constant for $x=\lambda/\eta>1$,
and after multiplying it by $\eta>0$ we obtain another positive constant 
equal to $\lambda-\eta\ln(e\lambda/\eta)$ that depends only on $\lambda$ and $\eta$.
Finally, we multiply this constant by $p>0$ to obtain the final constant $c>0$
dependent only on $\lambda,\eta,p$.

The same result for packets of length $\lmax$ can be proved by replacing $p$ by $q=1-p$ in the above analysis. 

Statement 1(b) is proved analogously to the first one, by replacing $\eta$ by $\eta'$.
This is possible because the Chernoff bound for Poisson process has the same form regardless
whether the upper or the lower bound on the Poisson value is considered, 
c.f., \cite{chernoff}.
\end{proof}

Now we can show the following result.

\begin{theorem}
\sloppypar{Let $p < q$. Then, the relative throughput of any algorithm $\Alg$ is at most 
$\min \left\{ \max \left\{\lambda p\lmin, \frac{\rhoflr} {\rho+ \rhoflr} \right\}, \frac{\rhoflr}{\rho} \right\}.$}
\end{theorem}

\newcommand{\lminarr}[1]{a^{\min}_{#1}}
\newcommand{\lmaxarr}[1]{a^{\max}_{#1}}
\newcommand{\Almins}{s^{\mathit{alg}}}
\newcommand{\Olmaxs}{r^{\mathit{off}}}
\newcommand{\Olmins}{s^{\mathit{off}}}

\begin{proof}
The claim has two cases. In the first case,  
$\lambda p\lmin \geq \frac{\rhoflr}{\rho}$. In this case, the upper bound of $\frac{\rhoflr}{\rho}$ is provided by Theorem~\ref{theo:abs-ub}.
In the second case $\lambda p\lmin < \frac{\rhoflr}{\rho}$. For this case, define two constants $\eta,\eta'$ such that $0 < \eta < \lambda < \eta'$
and $\eta' p < \eta q$. Observe that these constants always exist. Then, we prove that the relative throughput of any algorithm $\Alg$ in this case is at most 
$\max \left\{\eta' p\lmin, \frac{\rhoflr} {\rho+ \rhoflr} \right\}.$

Let us introduce some notation. We use $\lminarr{t}$ and $\lmaxarr{t}$ to denote the number of $\lmin$ and $\lmax$ packets, respectively, injected up to time $t$. 
%(When convenient we will use a time interval $\tau$ instead of a time $t$.) 
%
Let $\Olmaxs_t$ and $\Olmins_t$ be the number of $\lmax$ and $\lmin$ packets respectively, successfully transmitted by $\OFF$ by time $t$. Similarly, let $\Almins_t$ be the number of $\lmin$ packets transmitted by algorithm $\Alg$ by time $t$. Observe that $\Almins_t \geq \Olmaxs_t \geq \lfloor \frac{\Almins_t}{\rhoflr} \rfloor$.

Let us consider a given execution and the time instants at which the queue of $\OFF$ is empty of $\lmin$ packets in the execution. We consider two cases.

\noindent
Case 1: For each time $t$, there is a time $t' > t$ at which $\OFF$ has the queue empty of $\lmin$ packets. Let us fix a value $\delta>0$ and define time instants $t_0, t_1, \ldots$ as follows. $t_0$ is the first time instant no smaller than $\lmin$ at which $\OFF$ has no $\lmin$ packet and such that $\lminarr{t_0} > \lmax$. Then, for $i>0$, $t_i$ is the first time instant no smaller than $t_{i-1} + \delta$ at which $\OFF$ has no $\lmin$ packets. The relative throughput at time $t_i$ can be bounded as 
$$
T_{\Alg}(A,E,t_i) \leq \frac{\Almins_{t_i} \lmin} {\Olmaxs_{t_i} \lmax + \lminarr{t_i} \lmin}
\leq \frac{\Almins_{t_i} \lmin} {\lfloor \frac{\Almins_{t_i}}{\rhoflr} \rfloor \lmax + \lminarr{t_i} \lmin}
\leq \frac{\Almins_{t_i} \lmin} {( \frac{\Almins_{t_i}}{\rhoflr} -1) \lmax + \lminarr{t_i} \lmin}.
$$
This bound grows with $\Almins_{t_i}$ when $\lminarr{t_i} > \lmax$, which leads to a bound on the relative throughput as follows.
$$
T_{\Alg}(A,E,t_i) \leq \frac{\lminarr{t_i} \lmin} { \lminarr{t_i} (\frac{\lmax}{\rhoflr} + \lmin) - \lmax} = \frac{\lminarr{t_i} \rhoflr} { \lminarr{t_i} (\rho+ \rhoflr) - \rho \rhoflr}.
$$
Which as $i$ goes to infinity yields a bound of $\frac{\rhoflr} {\rho+ \rhoflr}$.

\noindent
Case 2:
There is a time $t_*$ after which $\OFF$ never has the queue empty of $\lmin$ packets. 
Recall that for any $t \geq \lmin$,
from Lemma \ref{lem:chernoff}, we have that the number of $\lmin$ packets injected by time $t$ satisfy $\lminarr{t} > \eta' p t$ with probability at most $\exp(-c' t)$ and the injected $\max$ packets satisfy $\lmaxarr{t} < \eta q t$ with probability
at most $\exp(-c t)$. 
By the assumption of the theorem and the definition of $\eta$ and $\eta'$, $\eta' p < \eta q$. Let us define $t^*= 1/ (\eta q - \eta' p)$.
Then, for all $t \geq t^*$ it holds that $\lmaxarr{t} \geq \lminarr{t}+1$, with probability at least $1- \exp(-c' t) - \exp(-c t)$.
If this holds, it implies that $\OFF$ will always have $\lmax$ packets in the queue.

Let us fix a value $\delta>0$ and define $t_0=\max(t_*,t^*)$, and the sequence of instants $t_i = t_0 + i \delta$, for $i=0,1,2,\ldots$.
By the definition of $t_0$, at all times $t > t_0$ $\OFF$ is successfully transmitting packets.
Using Lemma \ref{lem:chernoff}, we can also claim that in the interval $(t_0,t_i]$ the probability that more than $\eta' p i \delta$ packets $\lmin$ are injected is 
no more than $\exp(-c'' i \delta)$. 

With the above, the relative throughput at any time $t_i$ for $i \geq 0$ can be bounded as
$$
T_{\Alg}(A,E,t_i) \leq \frac{ (\lminarr{t_0} + \eta' p \cdot i \delta) \lmin} {\Olmaxs_{t_0} \lmax + \Olmins_{t_0} \lmin + i \delta}
$$
with probability at least $1- \exp(-c t_i) - \exp(-c' t_i) - \exp(-c'' t_i)$. Observe that as $i$ goes to infinity the above bound converges to $\eta' p \lmin$, while the probability converges exponentially fast to 1.
\end{proof}

\subsection{Lower Bound and Algorithm \CRalg}

In this section we consider algorithm \CRalg (stands for Conditional \Ralg), which 
builds on algorithm \Ralg presented in Section~\ref{sec:ralg}, in order to solve packet scheduling
in the setting of stochastic packet arrivals. The algorithm, depending on the
arrival distribution, either follows the $\SL$ policy (giving
priority to $\lmin$ packets) or algorithm \Ralg. More precisely,
algorithm \CRalg acts as follows:
\begin{center}
If $\lambda p \lmin > \frac{\rhoflr}{2\rho}$ then algorithm $\SL$ is run, otherwise
algorithm \Ralg is executed.
\end{center}

Then we show the following:

\begin{theorem}
\label{th:Cralg}
The relative throughput of algorithm \sloppy{\CRalg is not smaller than  
$\frac{\rhoflr}{\rho + \rhoflr}$}
for $\lambda p \lmin \le \frac{\rhoflr}{2\rho}$,
and not smaller than
$\min\left\{\lambda p \lmin,\frac{\rhoflr}{\rho}\right\}$
otherwise.
\end{theorem}

\begin{proof}
We consider three complementary cases.

{\bf Case $\lambda p \lmin \le \frac{\rhoflr}{2\rho}$.}
In this case algorithm \CRalg runs algorithm \Ralg,
achieving, per Theorem~\ref{thm:Ralg}, relative throughput of at least 
$\frac{\rhoflr}{\rho+\rhoflr}$ 
under {\em any} error pattern.

{\bf Case $\frac{\rhoflr}{2\rho}\le\lambda p \lmin\le 1$.}
Our goal is to prove that the relative throughput is not smaller than $\min\left\{\eta p \lmin,\frac{\rhoflr}{\rho}\right\}$,
for any $\eta$ satisfying $\lambda/2<\eta<\lambda$.
Considering such an $\eta$ 
we can make use of Lemma~\ref{l:stoch-lmin}
with respect to $\lambda,\eta,p$.
The relative throughput compares the behavior of algorithm \CRalg, which is simply $\SL$ in this case, with \OPT for each execution. Hence, for the purpose of the analysis we introduce the following modification in every execution:
we remove all periods in which \OPT is not transmitting any packet.
By ``removing'' we understand that we count time after removing the \OPT-unproductive periods and ``gluing'' the remaining periods so that they form one time line. In the remainder of the analysis
of this case we consider these modified executions with modified time lines 
and whenever we need to refer to the ``original'' time line we use the notion of {\em global time}.

For any positive integer $i$, we define time points
$t_i=i \cdot \lmax$.
Consider events $S_i$, for positive integers $i$, defined as follows:
the number of packets arrived by time $t_i$ (on the modified time line of the considered execution)
is at least $t_i\eta p$. By Lemma~\ref{l:stoch-lmin} and the fact that time $t$ on the modified time line
cannot occur before the global time $t$, 
there is a constant $c$ dependent only on $\lambda,\eta,p$ such that for any $i$:
the event $S_i$ holds with probability at least $1-\exp{(-ct_i)}$.

Consider an integer $j>1$ being a square of another integer.
We prove that
by time $t_j$, the relative throughput is at least 
\[
\min\left\{\eta p \lmin - \frac{\rhoflr\lmin}{t_j},
(1-1/\sqrt{j})\cdot \frac{\rhoflr}{\rho}\right\}
\]
with probability at least
$1-c' \exp{(-ct_{\sqrt{j}})}$, for some constant $c'>1$ dependent
only on $\lambda,\eta,p$.
To show this, consider two complementary scenarios that may happen at time $t_j$:
there are at least $\rhoflr$ pending packets of length $\lmin$, or otherwise.
It is sufficient to show the sought property separately in each of these two 
scenarios.

Consider the first scenario, when there are at least 
$\rhoflr$
pending packets of length $\lmin$
at time $t_j$.
\sloppy{With probability at least $1-c' \exp{(-ct_{\sqrt{j}})}$, for every $\sqrt{j}\le i\le j$ 
at least $t_{i}\eta p$ packets arrive by time $t_i$.}
This is because of the union bound of the corresponding events $S_i$
and the fact that 
$\sum_{i\ge\sqrt{j}} \exp{(-ct_i)} \le c' \cdot \exp{(-ct_{\sqrt{j}})}$
for some constant $c'>1$ dependent on $\lambda,\eta,p$ 
(note here that although $c'$ seems to depend also on $c$, 
$c'$ is still dependent only on $\lambda,\eta,p$ because $c$
is a function of these three parameters as well).
Consider executions in  $\bigcup_{i=\sqrt{j}}^j S_i$.
Using induction on $i$, if follows that for these executions
for every $\sqrt{j}\le i\le j$ the following invariant holds:
at least 
$t_i \eta p - \rhoflr$ 
packets of length $\lmin$ have been successfully transmitted by time $t_i$
or in the time interval $[t_i,t_{i+1}]$ at least $\rhoflr$ packets of length $\lmin$ are 
successfully transmitted (i.e., these successful transmissions end in the interval $[t_i,t_{i+1}]$).
The inductive proof of this invariant follows directly from the specification of algorithm \CRalg
(recall that it simply runs algorithm $\SL$ in the currently considered case) 
and from the definition of the modified execution and time line.
Let $i^*$ denote the largest $i\in [\sqrt{j},j]$ satisfying the following condition:
there are less than $\rhoflr$ packets of length $\lmin$ pending in time $t_i$;
if such an $i$ does not exist, we set $i^*=-1$.
Consider two sub-cases.

\noindent
{\em Sub-case $i^*\ge \sqrt{j}$ .}
If follows from the invariant and the definition of $i^*$ that by time $t_{i^*}$ there are at least 
$t_i \eta p - \rhoflr$ successfully transmitted packets of length $\lmin$,
and in each interval $[t_i,t_{i+1}]$, for $i^*\le i<j$, at least $\rhoflr$ packets of length $\lmin$
finish their successful transmission.
Therefore, by time $t_j$ the total length of packets (of length $\lmin$) successfully transmitted by 
algorithm \CRalg is at least 
\[
(t_{i^*} \eta p - \rhoflr)\lmin + \frac{t_j- t_{i^*}}{\lmax}\cdot \rhoflr\lmin
\ ,
\]
while the total length of successfully transmitted packets by \OPT by time $t_j$ is at most
$t_j$, by the definition of the modified execution and time line.
Therefore the relative throughput is at least
\[
\frac{(t_{i^*} \eta p - \rhoflr)\lmin + \frac{t_j- t_{i^*}}{\lmax}\cdot \rhoflr\lmin}{t_j}
\]
\[
\ge 
\min\left\{\frac{(t_{j} \eta p - \rhoflr)\lmin}{t_j},
\frac{\frac{t_j- t_{\sqrt{j}}}{\lmax}\cdot \rhoflr\lmin}{t_j}\right\}
\]
\[
=
\min\left\{\eta p \lmin - \frac{\rhoflr\lmin}{t_j},
(1-1/\sqrt{j})\cdot \frac{\rhoflr}{\rho}\right\}
\ .
\]
This converges to 
$\min\left\{\eta p \lmin,\frac{\rhoflr}{\rho}\right\}$ with $j$ going to infinity.

\noindent
{\em Sub-case $i^*<\sqrt{j}$ .}
In this sub-case we have, by definition of \mbox{$i^*<\sqrt{j}$}, that
at every time $t_i$, where $\sqrt{j}\le i\le j$, there are at least $\rhoflr$ pending packets of length $\lmin$.
Consequently, by the specification of the algorithm,
in each interval $[t_i,t_{i+1}]$, for $\sqrt{j}\le i<j$, at least $\rhoflr$ packets of length $\lmin$
finish their successful transmission.
Therefore, by time $t_j$ the total length of packets (of length $\lmin$) successfully transmitted by 
algorithm \CRalg is at least 
\[
\frac{t_j- t_{\sqrt{j}}}{\lmax}\cdot \rhoflr\lmin
\ ,
\]
while the total length of successfully transmitted packets by \OPT by time $t_j$ is at most
$t_j$, by the definition of the modified execution and time line.
Therefore the relative throughput is at least
\[
\frac{\frac{t_j- t_{\sqrt{j}}}{\lmax}\cdot \rhoflr\lmin}{t_j}
%\]
%\[
=
(1-1/\sqrt{j})\cdot \frac{\rhoflr}{\rho}
\ ,
\]
and it converges to 
$\frac{\rhoflr}{\rho}$ with $j$ going to infinity.
This completes the analysis of the sub-cases.

Finally, it is important to notice that the final converge of the ratio, with $j$ going to infinity,
in both sub-cases gives a valid bound on the relative throughput, since 
the subsequent ratios hold with probabilities approaching $1$ exponentially fast (in $j$),
i.e., with probabilities at least  $1-c' \exp{(-ct_{\sqrt{j}})}$, where $c$ and $c'$ are positive
constants dependent only on $\lambda,\eta,p$.
The minimum of the two relative throughputs, coming from the sub-cases,
is $\min\left\{\eta p \lmin,\frac{\rhoflr}{\rho}\right\}$, as desired
and therefore the relative throughput is at least
%This finishes the proof that the relative throughput in this case cannot be smaller than
$\min\left\{\lambda p \lmin,\frac{\rhoflr}{\rho}\right\}$
in this case.

{\bf Case $\lambda p\lmin> 1$.}
In this case we simply observe that we get at least the same relative throughput as
in case $\lambda p \lmin=1$, because we are dealing with executions saturated 
with packets of length $\lmin$ with probability converging to $1$ exponentially fast.
(Recall that we use the same algorithm $\SL$ in the specification of \CRalg, both for 
$\lambda p \lmin=1$ and for $\lambda p \lmin>1$.)
\sloppy{Consequently, the relative throughput in this case is at least
$\min\left\{\eta p \lmin,\frac{\rhoflr}{\rho}\right\}$, for any $\lambda/2<\eta<\lambda$,
and therefore it is at least
%cannot be smaller than
$\min\left\{\lambda p \lmin,\frac{\rhoflr}{\rho}\right\}
\ge
\min\left\{1,\frac{\rhoflr}{\rho}\right\}
=
\frac{\rhoflr}{\rho}$.}
\end{proof}

\remove{
A sketch of the proof is as follows. We provide further details in the Appendix.
We break the analysis into cases according to the probability of $\lmin$ packet arrivals and consider the time line of executions ignoring any OPT-unproductive periods.\vspace{.2em} 

\noindent
{\bf Case $\lambda p \lmin \le \frac{\rhoflr}{2\rho}$.}
In this case algorithm \CRalg runs algorithm \Ralg,
achieving, per Theorem~\ref{thm:Ralg}, relative throughput of at least 
$\frac{\rhoflr}{\rho + \rhoflr}$
under {\em any} error pattern.\vspace{.2em} 

\noindent
{\bf Case $\frac{\rhoflr}{2\rho}\le\lambda p \lmin\le 1$.}
It can be proved that the relative throughput is not smaller than $\min\left\{\eta p \lmin,\frac{\rhoflr}{\rho}\right\}$,
for any $\eta$ satisfying $\lambda/2<\eta<\lambda$. To prove it, 
we consider time points $t_i$ being multiples of $\lmax$ and show that with high probability, at those points there have already arrived at least $t_i\eta p$ packets. Using this property, we show that the relative
throughput at time $t_j$ is at least
$\min\left\{\eta p \lmin - \frac{\rhoflr\lmin}{t_j},
(1-1/\sqrt{j})\cdot \frac{\rhoflr}{\rho}\right\}$
with probability at least $1-c' \exp{(-ct_{\sqrt{j}})}$, for some constant $c,c'>0$ dependent
only on $\lambda,\eta,p$.
It follows that if $j$ grows to infinity, we obtain the desired relative throughput.\vspace{.2em}

\noindent
{\bf Case $\lambda p\lmin> 1$.}
In this case we simply observe that we get at least the same relative throughput as
in case $\lambda p \lmin=1$, because we are dealing with executions saturated 
with packets of length $\lmin$ with probability converging to $1$ exponentially fast.
(Recall that we use the same algorithm $SL$ in the specification of \CRalg, both for 
$\lambda p \lmin=1$ and for $\lambda p \lmin>1$.)
\sloppy{Consequently, the relative throughput in this case is at least
$\min\left\{\eta p \lmin,\frac{\rhoflr}{\rho}\right\}$, for any $\lambda/2<\!\eta<\!\lambda$,
and thus 
it is at least
$\min\left\{\lambda p \lmin,\frac{\rhoflr}{\rho}\right\}
\ge
\min\left\{1,\frac{\rhoflr}{\rho}\right\}
=
\frac{\rhoflr}{\rho}$.}\vspace{.3em}

\noindent Combining the three cases, we get the claimed result (see Appendix~\ref{a:stoc-inst} for details).
}

Observe that if we compare the upper bounds on relative throughput shown in the previous subsection
with the lower bounds of the above theorem, then we may conclude that in the case where $\rho$
is an integer, algorithm $\CRalg$ is optimal (wrt relative throughput). In the case where $\rho$
is not an integer, there is a small gap between the upper and lower bound results.

\section{Conclusions}
\label{sec:conclusions}

This work was motivated by the following observation regarding the system of dynamic packet arrivals with errors: scheduling packets of same length is relatively easy and efficient in case of instantaneous feedback, 
but extremely inefficient in case of deferred feedback.
We studied scenarios with two different packet lengths, developed efficient algorithms, and proved upper and lower bounds for relative throughput
in average-case (i.e., stochastic) and worst-case (i.e., adversarial) online packet arrivals.
These results demonstrate that exploring instantaneous feedback mechanisms (and developing
more effective implementations of it) has the potential 
to significantly increase the performance of communication systems.

Several future research directions emanate from this work. Some of them concern the exploration of 
variants of the model considered, for example, assuming that packets that suffer errors are not retransmitted (which applies when Forward Error Correction~\cite{CED} is used), considering
packets of more than two lengths, or assuming bounded buffers. Other lines of work deal with adding QoS requirements to the problem, such as requiring fairness in the transmission of the packets from different flows or imposing deadlines to the packets.
In the considered adversarial setting, it is easy to see that even an omniscient offline solution cannot achieve stability: for example, the adversary could prevent any packet from being transmitted correctly. Therefore, an interesting extension of our work would be to study conditions (e.g., restrictions on the adversary) under 
which an online algorithm could maintain stability, and still be efficient with respect to relative throughput.
Finally, we believe that the definition of relative throughput as proposed here can be adapted, possibly in a different context, to other metrics and problems.

\bibliographystyle{plain}
\bibliography{references}

\newpage
\appendix
\section*{APPENDIX}

\section{NP-hardness}
\label{sec:np-hard}

We prove the NP-hardness of the following problem, defined for a single link.

INSTANCE (Throughput Problem):  Set $X$ of packets, for each packet $x \in X$ a length $l(x) \in \mathbb{N}^+$, an arrival time $a(x) \in \mathbb{Z}^0$, a sequence of time instants $0=T_0 < T_1 <  T_2 < \cdots < T_k$, 
$T_i \in \mathbb{N}^0$, so that the
link suffers an instantaneous error at each time $T_i$, $i \in [1,k]$
(in other words, at each time $T_i$, any packet transmitted over the link is corrupted).

QUESTION: is there a schedule of $X$ so that error-free packets of total length $T_k$ are transmitted by time $T_k$ over the link?

\begin{theorem}
The Throughput Problem is NP-hard.
\end{theorem}
\begin{proof}
We use the 3-Partition problem which is known to be an NP-hard problem. 

INSTANCE:  Set $A$ of $3m$ elements, a bound $B\in\mathbb{N}^+$ and, for each $a\in A$, a size $s(a)\in\mathbb{N}^+$ such that $B/4<s(a)<B/2$ and $\sum_{a\in A} s(a) = mB$. 

QUESTION: can $A$ be partitioned into $m$ disjoint sets $\{A_1,A_2,\dots, A_m\}$ such that, for each $1\leq i\leq m$, $\sum_{a\in A_i} s(a) = B$?

We reduce the 3-Partition problem to the Throughput Problem, defined for a single link. 
The reduction is by setting $X=A$, $l()=s()$, $a()=0$, $k=m$, and $T_i=i B$ for $i \in [1,k]$.
If the answer to 3-Partition is affirmative, then for the Throughput Problem there is a way to schedule (and transmit) the packets in $X$ in subsets $\{X_1,X_2,\dots,X_m\}=\{A_1,A_2,\dots,A_m\}$, so that all the packets in $A_i$ can be transmitted over the link in the interval $[T_{i-1},T_i]$. Furthermore, since $\sum_{a\in A_i} s(a) = \sum_{x\in X_i} l(x) = B$, and $T_i - T_{i-1}=B$,
the total length of packets transmitted by time $T_k$ is $T_k$.

The reverse argument is similar. If there is a way to schedule packets so that the total packet length transmitted by time $T_k$ is $T_k$, in each interval between two error events on the link  
%(re)starts of the processor 
there must be exactly $B$ bytes of packets transmitted. Then, the packets can be partitioned into subsets of total length $B$ each. 
This implies the partition of~$A$.
\end{proof}

\section{Deferred Feedback}
\label{sec:deferred}

In this section we study the relative throughput of any algorithm under the deferred feedback mechanism. As described in Section~\ref{s:intro}, with this mechanism the sending station is notified about a packet having been corrupted by an error only after the transmission of the packet is completed. 
Here we assume that all packets have the same length $\ell$. We show that even in this case no algorithm can achieve positive throughput.

\subsection{Adversarial Arrivals}
\label{sec:deferred1}
In order to prove the upper bound on throughput, the packets arrive frequently enough so that there are always packets ready.
The algorithm will then greedily send a train of packets.
The adversary injects bit errors at a distance of exactly $\ell$ so that each error hits a different packet, 
and hence the algorithm cannot successfully complete any transmission
(that is, it cannot transmit non-corrupted packets). 
At the same time, an offline algorithm \OFF is able to send packets in each interval of length $\ell$ without errors. This argument leads to the following theorem:

\begin{theorem} \vspace{-.5em}
No packet scheduling algorithm \Alg can achieve a relative throughput larger than $0$
under adversarial arrivals in the deferred feedback model, even with one packet length.
\end{theorem}

\subsection{Stochastic Arrivals}
\label{sec:deferred2}

Let us consider now stochastic arrivals. We show that also in this case the upper bound on the relative throughput is $0$.

\begin{theorem}
No packet scheduling algorithm \Alg can achieve a relative throughput larger than $0$
under stochastic arrivals in the deferred feedback model, even with one packet length.
\end{theorem}

\begin{proof}
As described in Section~\ref{sec:model}, we assume that packets arrive at a rate $\lambda$. Here we assume that all packets have the same length $\ell$.
Observe that if $\lambda \ell < 1$ there are many times when there is no packet ready to be sent and the link will be idle. In any case, the adversary can inject errors following the next rule: inject an error in the middle point of each packet sent by \Alg. Applying this rule, no packet sent by \Alg is received without errors. However,
between two errors there is at least $\ell$ space (even if packets are contiguous) and the offline algorithm \OFF can send a packet. The conclusion is that \OFF is able to successfully send at least one packet between two attempts of \Alg, while \Alg cannot complete successfully any transmission. This completes the proof.
\end{proof}

\section{Upper Bounds for Algorithms $\SL$ and $\LL$}
\label{subsec:SL-LL}

We prove upper bounds that suggest that algorithms $\SL$ (Shortest Length) and $\LL$ (Longest Length)
are not efficient. 
\sloppy{First, we show that $\SL$ cannot have relative throughput larger than
$\frac{1}{\rho + 1}$
under adversarial arrivals. We then show that algorithm $\LL$ is even worse, as its relative throughput cannot be more than $0$ even with stochastic arrivals.}

\begin{theorem}
Algorithm $\SL$ cannot achieve relative throughput larger than 
$\frac{1}{\rho + 1}$ under adversarial arrivals, even if there is a schedule that transmits all the packets.
\end{theorem}
\begin{proof}
The scenario works as follows. At time 0 two packets arrive, one of length $\lmax$ and one of length $\lmin$. 
$\SL$ schedules first the packet of length $\lmin$, and when it is transmitted, it schedules the packet of length $\lmax$. 
Meanwhile, an offline algorithm \OFF schedules first the packet of length $\lmax$. When it is transmitted, the adversary causes an error on the 
link, so $\SL$ does not transmit successfully the packet of length $\lmax$.
Now, $\SL$ only has one packet of length $\lmax$ in its queue (when this scenario is repeated will have several, but no packets of length $\lmin$). Hence, $\SL$ schedules this packet, while \OFF schedules the packet of length $\lmin$ that has in its queue. When \OFF completes the transmission of the $\lmin$ packet, 
the adversary causes an error on the link. This scenario can be repeated forever. In each instance, \OFF transmits one packet of length $\lmax$ 
and one of lenght $\lmin$, while $\SL$ only transmits one packet of length $\lmin$. Hence, the throughput achieved is $\frac{\lmin}{\lmax + \lmin} = \frac{1}{\rho + 1}$. 
Observe that at the end of each instance of the scenario the queue of \OFF is empty.
\end{proof}

We now show that the above upper bound also holds with stochastic arrivals under specific packet arrival rates.
%
%\added[ez]{It is important to look at the frequency of packet arrivals and not only at the long term.}

\begin{theorem}
$\forall \ep > 0, \exists \lambda, p, q$ such that algorithm $\SL$ cannot achieve a relative throughput larger than $\frac{1}{(1-\ep) \rho + 1} + \ep$.
\end{theorem}

\begin{proof}
Consider an execution of the $\SL$ algorithm. We define intervals $I_1, I_2, \dots, I_i$ as follows.
The first such interval, $I_1$, starts with the arrival of the first $\lmin$ packet.
Then, $I_i$ starts as soon as an $\lmin$ packet is in the queue of \SL after the end of interval $I_{i-1}$.
The length of each interval depends on whether \OFF has an $\lmax$ packet in its queue at the start of the interval or not. If it has an $\lmax$ packet, the length 
of the interval is $|I_i| = \lmin + \lmax$, and we say that we have a \emph{long} interval. If it does not, the length is $|I_i| = \lmin$ and the interval is called \emph{short.}

Between intervals the adversary injects frequent errors, so \SL cannot transmit any packet.
In every interval $I_i$, $\SL$ starts by scheduling an $\lmin$ packet. In a short interval, \OFF sends an $\lmin$ packet, followed by an error injected by the adversary.
Hence, in a short interval both \SL and \OFF successfully transmit one $\lmin$ packet.
In a long interval, \OFF sends an $\lmax$ packet, after which the adversary injects an error. (Up to that point $\SL$ has been able to complete the transmission of
one or more $\lmin$ packets, but no $\lmax$ packet.)  After the error, $\OFF$ sends an $\lmin$ packet (which is available since beginning of the interval) after which continuous errors will be injected by the adversary until the next interval. Hence, in a long interval \OFF successfully transmits one $\lmin$ packet and one $\lmax$ packet, while \SL transmits only $\lmin$ packets. This implies that in both types of intervals \OFF is transmitting useful packets during the whole interval.

Let us denote by $s_k$ the total length of the intervals $I_1, I_2, \ldots, I_k$, i.e., $s_k=\sum_{i=1}^k |I_i|$. Observe that the total number of $\lmin$ packets that
arrive up to the end of interval $I_k$ is bounded by $k$ (that accounts for the $\lmin$ packet in the queue of \SL at the start of each interval) plus the packets that arrive in the intervals. From Lemma~\ref{l:stoch-lmin}, we know that there is a constant $\eta' > \lambda$ and a constant $c' > 0$ which depends only on $\eta',\lambda$ and $p$, such that the number of $\lmin$ packets that arrive in the intervals is at most $\eta'p s_k$ with probability at least $1-e^{-c' s_k}$.

Let $T_k$ be the throughput of \SL at the end of interval $I_k$. From the above, we have that $T_k$ is bounded as
$$
T_k \leq \frac{\lmin (k+ \eta'p s_k)}{s_k} = \frac{\lmin k}{s_k} + \lmin \eta'p
$$
with probability at least $\pi_1(k)=1-e^{-c' s_k}$. Observe that in the above expression it is assumed that all $\lmin$ packets that arrive by the end of $I_k$ are
successfully transmitted by \SL. We provide now the following claim.
\\

\noindent
\emph{Claim:} Let us consider the first $x+1$ intervals $I_i$, for $x>1$. The number of long intervals is at least $(1-\delta)(1-e^{-\lambda q \lmin})x$ 
with probability at least $1-\exp(-\delta^2 (1-e^{-\lambda q \lmin})x/2)$, for any $\delta \in (0,1)$.

\noindent
\emph{Proof of claim:}
Observe that if an $\lmax$ packet arrives during interval $I_i$ then the next interval $I_{i+1}$ is long. We consider now the first $x$ intervals. Since each of these
intervals has length at least $\lmin$, some $\lmax$ packet arrives in the interval with probability at least $1-e^{-\lambda q \lmin}$ (independently of what happens in other intervals). Hence, using a Chernoff bound, the probability of having less than $(1-\delta)(1-e^{-\lambda q \lmin})x$ intervals among the $x$ first intervals 
in which $\lmax$ packets arrive is at most $\exp(-\delta^2 (1-e^{-\lambda q \lmin})x/2)$.
\qed
\\

From the claim, it follows that there are at least $(1-\delta)(1-e^{-\lambda q \lmin})(k-1)$ long intervals among the first $k$ intervals, with high probability. Hence, the value of $s_k$ is bounded as
\begin{eqnarray*}
s_k & \geq & (1-\delta)(1-e^{-\lambda q \lmin})(k-1)(\lmax + \lmin) + (k - (1-\delta)(1-e^{-\lambda q \lmin})(k-1)) \lmin \\
&=&  (1-\delta)(1-e^{-\lambda q \lmin})(k-1)\lmax + k \lmin
\end{eqnarray*}
with probability at least $\pi_2(k)=1-\exp(-\delta^2 (1-e^{-\lambda q \lmin})(k-1)/2)$.
Note that $T_K$ cannot be larger than 1.
Hence, the expected value of $T_k$ can be bounded as follows.
\begin{eqnarray*}
\E [T_k] \le \pi_1(k) \pi_2(k) \left( \frac{\lmin k}{(1-\delta)(1-e^{-\lambda q \lmin})(k-1)\lmax + k \lmin} + \lmin \eta'p \right) + (1 - \pi_1(k) \pi_2(k)).
\end{eqnarray*}
Since $\pi_1(k)$ and $\pi_2(k)$ tend to one as $k$ tends to infinity, we have that
\begin{eqnarray*}
\lim_{k \rightarrow \infty} \E [T_k] & \le & \frac{\lmin}{(1-\delta)(1-e^{-\lambda q \lmin})\lmax + \lmin} + \lmin \eta'p\\
&=& \frac{1}{(1-\delta)(1-e^{-\lambda q \lmin})\gamma + 1} + \lmin \eta'p.
\end{eqnarray*}
Hence, choosing $\eta'$, $p$, $q$, and $\delta$ appropriately, the claim of the theorem follows. (E.g., they must satisfy
$\lmin \eta'p \le \ep$ and $(1-\delta)(1-e^{-\lambda q \lmin}) \ge (1-\ep)$.)
\end{proof}

\remove{
The probability that there are more than $x$ arrivals of packets $\lmin$ in an interval $\Delta$ is at most $e^{-\lambda p\Delta}(\frac{e\lambda p\Delta}{x})^x = e^{2-2\ln 2 + 2\ln \lambda p\Delta - \lambda p\Delta}$.

$\SL$ algorithm will transmit as many $\lmin$ packets as possible, from the ones it has pending, in the period $\Delta$ and therefore assuming the best case scenario for it we calculate the expected relative throughput considering $x$ to be the actual $\lmin$ packet arrivals within interval $\Delta$ and $i$ being the ones sent:
\begin{eqnarray*}
\E[T_{SL}]
& = &
\sum\limits_{x=0}^{\infty} x Pr[T_{SL} = x] \\
& = &
\sum\limits_{i=1}^{\infty} \frac{i\lmin}{\lmin + \lmax} * \frac{e^{-\lambda p\Delta}(\lambda p\Delta)^i}{i!} \\
& = &
\frac{\lmin\lambda p\Delta }{\lmin + \lmax} \sum\limits_{i=1}^{\infty} \frac{e^{-\lambda p\Delta} (\lambda p \Delta)^{i-1} }{(i-1)!} \\
& = &
\frac{\lambda p\Delta}{\rho + 1} \sum\limits_{i=0}^{\infty} \frac{e^{-\lambda p\Delta} (\lambda p \Delta)^i}{i!} \\
& = &
\frac{\lambda p\Delta}{\rho + 1}
\end{eqnarray*}
}

%\begin{theorem}
%Under stochastic arrivals with $\lambda p(\lmin + \lmax) \ll 1$ and $q \gg 2p$, algorithm $\SL$ cannot achieve relative throughput larger than $\frac{1}{\rho + 1}$.
%\end{theorem}

\remove{
\begin{proof}[Sketch]
First, let us consider a value of $\lambda p$ such that $\lambda p(\lmin + \lmax) \ll 1$. This, in combination with Lemma~\ref{l:stoch-lmin}, means that there is a constant $c>0$ that depends only on $\lambda, p$ and $\eta$, where $\eta = (1+\ep)\lambda$, such that for any time $t > \lmin + \lmax$, the number of $\lmin$ packets injected by time $t$ is at most $t\eta p$, with probability at least $1-e^{-ct}$. In other words, at most one $\lmin$ packet is injected at the sender with high probability every $\lmin + \lmax$ time.
Consider also $q \gg 2p$. This implies that $\lambda q \gg 2\lambda p$, which in combination with Lemma~\ref{l:stoch-lmin}, means that there are at least twice as many $\lmax$ than $\lmin$ packets injected at the sender by time $t$, with probability $1-e^{-ct}$, for a positive constant $c$ and $t>\lmin$.

Consider now the adversary of Section~\ref{sec:stocUB} and algorithm $\SL$. There are two types of phases again. First, when the phase starts with $\SL$ algorithm sending an $\lmin$ packet (as soon as it has one pending), $\OFF$ will send an $\lmax$. Observe that, by the previous arguments and specifically the fact that $q \gg 2p$, $\OFF$ will have an $\lmax$ packet pending in all such cases. By the time the $\lmax$ packet has been transmitted by $\OFF$, $\SL$ has attempted to send an $\lmax$ packet (provided it has no more $\lmin$ packets pending, which happens w.h.p.) but will fail due to the link error caused by the adversary as soon as $\OFF$ completes the $\lmax$ transmission. Note here, that the $\SL$ algorithm sends an $\lmax$ packet only when it has no $\lmin$ packets pending, and this is exactly what the adversary exploits. Observe also, that from the fact that $\lambda p(\lmin + \lmax) \ll 1$, this is actually the case because $\SL$ will have no more $\lmin$ packets pending by the time it completes transmitting the one, w.h.p.
Then, in the case that $\SL$ starts a phase by sending an $\lmax$ packet (again when no $\lmin$ packet is pending), $\OFF$ will send its pending $\lmin$ packet and an error will occur before the $\lmax$ packet is transmitted by $\SL$. After that, $\OFF$ does not have an $\lmin$ packet pending, and the adversary will introduce errors until the next $\lmin$ injection in order to create another phase of the previous type.

Hence, the throughput achieved is $\frac{\lmin}{\lmax + \lmin} = \frac{1}{\rho + 1}$ as claimed.
\end{proof}
}

\remove{
\begin{proof}
Algorithm $\SL$ attempts to transmit the packet of shortest length available. For the proof we are going to consider the adversary of Section~\ref{sec:stocUB}. Hence,
if there are packets of length $\lmin$, $\SL$ schedules one, that is correctly transmitted. After that (as we will show whp) it will have no other packet of length $\lmin$ but will have packets of length $\lmax$, and will hence schedule one. When the $\lmin$ packet is scheduled, algorithm $\OFF$ sends a packet of length $\lmax$ (we will show that such a packet is available with high probability) that is also correctly transmitted. Immediately after this packet is sent the adversary introduces an error that affects the
transmission of the $\lmax$ packet of $\SL$. After the error, $\SL$ still have no $\lmin$ packet in the queue (whp) and schedules again an $\lmax$ packet, while $\OFF$ schedules the $\lmin$ packet previously sent by $\SL$. The adversary allows this $\lmin$ packet to be sent without errors, but introduces errors in the $\lmax$ packet of $\SL$. 

Let us consider an arbitrary time $t$ and let us define constants $\eta$ and $\eta'$ such that $0 < \eta < \lambda < \eta'$ and $\eta' p(\lmin + \lmax) < 1$. From Lemma~\ref{l:stoch-lmin}, the number of $\lmin$ packets that arrive by time $t$ is at most $\eta^\prime pt$ with probability at least $1-e^{-ct}$. 

The expected relative throughput of an execution at time $t$, under a distribution of stochastic arrivals $\mathcal{A}$, is equal to the sum of expected throughputs across all possible arrival patterns and their probability of occurrence; $\E(T) = \sum\limits_{A \in \mathcal{A}}T(A) Pr[A]$. Hence, for a time $t_i$, at which the algorithm we consider is about to decide on the next packet to be sent, if the expected relative throughput is $T_{t_i} \leq B$ with a probability $Pr \geq 1-e^{-ct_i}$, it is also true that $1 > T_{t_i} \geq B$ with probability $Pr < e^{ct_i}$. At such a point, $\E(T_{t_i}) \leq B(1-e^{-ct_i}) + 1e^{-ct_i}$ and hence in the long run, the relative throughput becomes $T = \lim\limits_{t\rightarrow\infty} \E[T_{t_i}] = B$.

Now for the case of stochastic arrivals with $\lambda p(\lmin + \lmax) \ll 1$ and $q \gg 2p$, let us define an $\eta^\prime < \lambda$ such that $\eta^\prime p(\lmin + \lmax) < 1$. From Lemma~\ref{l:stoch-lmin}, the number of $\lmin$ packets that arrive by time $t_i$ is at most $\eta^\prime pt_i$ with probability at least $1-e^{-ct_i}$. Also, since $q > 2p \Rightarrow \lambda q > 2\lambda p$, which means that the number of $\lmax$ packets that arrive by time $t_i$ are at least twice as many as the $\lmin$ packets. Therefore, there have arrived at least $\eta^\prime qt_i$ packets of length $\lmax$ with probability at least $1-e^{-ct_i}$.

Considering the adversary of Section~\ref{sec:stocUB} and algorithm $\SL$, there are two types of phases. The first, $\SL$ algorithm starts sending an $\lmin$ packet while $\OFF$ sends an $\lmax$. Observe that, by the previous arguments $\OFF$ will have an $\lmax$ packet pending with probability at least $1-e^{-ct_i}$. By the time the $\lmax$ packet has been transmitted by $\OFF$, $\SL$ will have already made an attempt to send an $\lmax$ packet (provided it has no more $\lmin$ packets pending) but will fail due to the link error caused by the adversary as soon as $\OFF$ completes the $\lmax$ transmission. Note here, that the $\SL$ algorithm sends an $\lmax$ packet only when it has no $\lmin$ packets pending, and this is exactly what the adversary exploits. Observe also, that from the fact that $\lambda p(\lmin + \lmax) \ll 1$, this will be the case with probability $1-e^{-ct}$.
%, because $\SL$ will have no more $\lmin$ packets pending by the time it completes transmitting one.
In the second case, $\SL$ starts by sending an $\lmax$ packet (again when no $\lmin$ packet is pending) and $\OFF$ now sends its pending $\lmin$ packet. An error will occur before the $\lmax$ packet is transmitted by $\SL$ and after that, $\OFF$ will not have an $\lmin$ packet pending. Therefore, the adversary will introduce errors until the next arrival of an $\lmin$ packet in order to create another phase of the previous type.

As a result, the expected relative throughput achieved by time $t$ is at least $\frac{\lmin}{\lmax + \lmin} = \frac{1}{\rho + 1}$ with probability $1-e^{-ct}$ as claimed.
\end{proof}
}

\begin{theorem}
Algorithm $\LL$ cannot achieve relative throughput larger than $0$, even under stochastic arrivals.
\end{theorem}
\begin{proof}
The scenario is simple. The adversary blocks all successful transmissions (by placing errors at distance smaller than $\lmin$)
until at least two packets have arrived, one of length $\lmax$ and one of length $\lmin$. Algorithm $\LL$ schedules
a packet of length $\lmax$, while an offline algorithm \OFF schedules a packet of length $\lmin$. Once \OFF completes the transmission of this
packet, the adversary causes an error on the link, and hence $\LL$ does not complete the transmission of the $\lmax$ packet. Then,
again the adversary blocks successful transmissions until \OFF has at least one $\lmin$ packet pending. 
The scenario is repeated for ever; while \OFF will be transmitting successfully
all $\lmin$ packets, $\LL$ will be stuck on the unsuccessful transmissions of $\lmax$ packets.
Hence, the throughput will be $0$. 
\end{proof}

\section{Randomized Algorithms}
\label{sec:randomization}

So far we have considered deterministic solutions. 
In many cases, randomized
solutions can obtain better performance. As we argue in this section,
this is not the case for the problem considered in this work. 

Let us first indicate how the model and the definition of relative throughput must be extended to the case of randomized algorithms. 
We assume that the adversary knows the algorithm and the history of the random choices made by the algorithm until the current point in time, but it does not know the future random choices made by the algorithm.

Regarding the relative throughput, and following the terminology of Section~\ref{sec:model}, in the case of randomized algorithms, an adversarial error-function $E$ has 
three arguments: an arrival pattern $A$, a string of values of random bits $R$, and time $t$.
The output of $E(A,R,t)$ is a set of errors until time $t$ based on the execution of
a given randomized algorithm with the values of random bits taken from $R$
under an adversarial pattern $A$ by round $t$.

For arrival pattern $A$, adversarial error-function $E$, string of random bits $R$ and time $t$, 
we define the \emph{relative throughput $T_{\Alg}(A,E,R,t)$ of a randomized
algorithm $\Alg$ by time $t$} as follows: 
\[
T_{\Alg}(A,E,R,t) = \frac{L_{\Alg}(A,E,R,t)}{L_{\OPT}(A,E,R,t)}
\ .
\]
$T_{\Alg}(A,E,R,t)$ is defined as 1 if $L_{\Alg}(A,E,R,t)=L_{\OPT}(A,E,R,t)=0$.

(Note that  \OPT is not randomized, but since the error-function $E$ depends on
the random choices of the algorithm, this has a direct effect on the performance of \OPT.)

We define the {\em relative throughput} of algorithm $\Alg$
in the adversarial arrival model as follows:
\[
T_{\Alg} = \inf_{A \in {\mathcal A}, E  \in {\mathcal E}} \lim_{t \rightarrow \infty} 
\E_{R\in{\mathcal R}} [T_{\Alg}(A,E,R,t)]
\ ,
\]
where $A$ is understood as a function of $R$ and $t$, 
and ${\mathcal R}$ is a distribution of all possible strings of random bits
used by the algorithm.
In the stochastic arrival model the relative throughput needs to take into account the random distribution of
arrival patterns in ${\mathcal A}$ (they are not functions now, as they do not depend
on the adversary), and it is defined as follows:
\[
T_{\Alg} = \inf_{E  \in {\mathcal E}} \lim_{t \rightarrow \infty} 
\E_{A \in {\mathcal A}, R\in{\mathcal R}} [T_{\Alg}(A,E,R,t)]
\ .
\]

Now, looking at the analyses of the upper bounds for deterministic algorithms with deferred feedback (Section~\ref{sec:deferred}) 
and with instantaneous feedback (under adversarial arrivals, Section~\ref{sec:advUB}, and 
stochastic arrivals, Section~\ref{sec:stocUB}), it is not difficult to see
that the derived bounds hold also for randomized algorithms. The main observation that leads to
this conclusion is the following: The adversarial error and arrival patterns defined in the analyses
are reactive, in the sense that the adversary that controls them does not need to know the future (and in particular the future random bits of the algorithm
) and makes its decisions only by looking at the system's history. In other words, when a given algorithm decides in a given phase what
packet length to transmit, the adversary reacts adaptively on the specific choice, regardless 
of whether this choice was done deterministically or by flipping a coin. This leads to the conclusion
that randomized solutions cannot yield better results (wrt relative throughput) for the considered packet
scheduling problem.

\end{document}